\documentclass[letterpaper,11pt]{article}
\usepackage{fullpage}
\usepackage{amsmath,amsthm, amssymb}
\usepackage{color}		
\usepackage{epsfig}
\usepackage{graphicx}
\usepackage{hyperref}
\usepackage[all]{hypcap}

\newcommand{\Integ}{{\mathbb Z}}

\newcommand{\Real}{{\mathbb R}}
\newcommand{\NP}{{\textrm NP}}
\newcommand{\M}{\mathcal{M}}
\newcommand{\I}{\mathcal{I}}

\newtheorem{theorem}{Theorem}[section]
\newtheorem{corollary}[theorem]{Corollary}
\newtheorem{lemma}[theorem]{Lemma}
\theoremstyle{definition}
\newtheorem{definition}[theorem]{Definition}
\newtheorem{alg}{Algorithm}

\hypersetup{
pdfauthor = {Jos\'e\ A.\ Soto},
pdftitle = {A simple PTAS for Weighted Matroid Matching on Strongly Base Orderable Matroids},
}

\begin{document}

\title{A simple PTAS for Weighted Matroid Matching on Strongly Base Orderable Matroids\footnote{This work was partially supported by NSF
  contract CCF-0829878 and by ONR grant N00014-11-1-0053.}\, \footnote{Extended abstract of this article to appear in the proceedings of the Sixth Latin-American Algorithms, Graphs and Optimization Symposium (LAGOS) 2011.}}
\author{Jos\'e A. Soto\thanks{MIT,
Dept.~of Math., Cambridge, MA 02139. \texttt{jsoto@math.mit.edu}.}}
\date{}
\maketitle

\begin{abstract}
We give a simple polynomial time approximation scheme for the weighted matroid matching problem on strongly base orderable matroids. We also show that even the unweighted version of this problem is NP-complete and not in oracle-coNP.
\end{abstract}

\section{Preliminaries}
We assume familiarity with basic graph and matroid theory. For reference, see Schrijver's book~\cite{Schrijver-book}. Given a matroid $\M=(S,\I)$ and a graph $G=(S,E)$, we say that a matching $F \subseteq E$ is  \emph{feasible} for $\M$ if the set $\bigcup F \subseteq S$ of vertices covered by $F$ is independent in $\M$. Given a weight function, $w\colon E \to \Real^+$, the \emph{weighted matroid matching problem} consists on finding a feasible matching of maximum weight. This problem is a common generalization to both the weighted graphic matching and the weighted matroid intersection problems.

Even for the unweighted version (where every edge has unit weight), the previous problem is not polynomially solvable when the matroid is given by an independence oracle~\cite{JensenK82,Lovasz81}: there are instances where checking that no feasible matching of a given size exists requires an exponential number of oracle queries. These instances can be modified to show \NP-completeness~\cite{Schrijver-book}.

For the unweighted version, Lov\'asz~\cite{Lovasz81} has given a polynomial time algorithm for linear matroids. Recently, Lee et al.~\cite{LeeSV10} have developed a PTAS for general matroids using local search techniques. For arbitrary weights, the situation is less developed. We only have polynomial time algorithms for the case where the matroid is a gammoid~\cite{tong1982solving}. An open problem of Lee et al.~\cite{LeeSV10} is to get a PTAS for the weighted version on general matroids.

In this note we focus on a special class of matroids containing gammoids and not comparable with linear matroids. A matroid $\M=(S,\I)$ is \emph{strongly base orderable} if for every pair of bases $I$~and~$J$, there is a bijection $\pi\colon I \to J$ such that for all $K \subseteq I$, the set $\pi(K) \cup (I\setminus K)$ is also a base of the matroid. On the positive side, we give a local search based PTAS for the weighted matroid matching problem on strongly base orderable matroids. On the negative side, we show that even the unweighted version of this problem is NP-complete and not in oracle co-NP. The negative results follow since the matroids used to prove the analogue hardness results for matroid matching are strongly base orderable.

An important special case of the weighted matroid matching problem is the \emph{weighted matroid parity problem}. In this version, the set $E$ of edges (also denoted as \emph{pairs} or \emph{mates}) is itself a matching. The following reduction~\cite{LeeSV10} shows that both problems are equivalent. Given an instance $(\M,G,w)$ of the weighted matroid matching problem, define a new set $S'$ of elements containing as many copies of every vertex $s$ as its degree in $G=(S,E)$. Construct a new graph $G'=(S',E')$ by creating for every edge $\{u,v\} \in E$, one edge of the same weight between a copy of $u$ and a copy of $v$ in $S'$, in such a way that the new collection of edges $E'$ is a matching. Let $\M'=(S',\I')$ be the matroid in which a set $I' \subseteq S'$ is independent if it has at most one copy of each element of $S$ and the respective set $I$ of original elements is independent in~$\M$. It is easy to check that the previous construction gives an approximation preserving reduction and that if $\M$ is strongly base orderable then so is the matroid $\M'$. The last statement follows since strong base orderability of matroids is closed under taking parallel extensions \cite[Section 42.6c]{Schrijver-book}.

\section{A local search PTAS for Weighted Matroid Parity}
Consider a matroid $\M=(S,\I)$, a partition $E$ of $S$ into $n$ pairs and a weight function $w\colon E\to \Real^+$. A folklore result (see, e.g.~\cite{LeeSV10}) is that the \emph{greedy solution} obtained by processing the pairs in decreasing order of weight, and selecting only those pairs that keep the current solution feasible, is a 2-approximation for the weighted matroid parity problem.

\begin{definition} An \emph{$s$-move} for a feasible set of pairs $A$ is a choice of pairs $\overline{A} \subseteq A$ and $\overline{B} \subseteq E\setminus A$, with $|\overline{A}|+|\overline{B}|\leq s$, such that  $D=A \Delta (\overline{A} \cup \overline{B})$ is feasible. The value $w(\overline{B})-w(\overline{A})$ is the \emph{gain of the move}. The move \emph{improves $A$} if its gain is positive (equiv., if $w(D) > w(A)$). For a given $\delta>0$, the move \emph{$\delta$-improves $A$} if its gain is at least $\delta w(A)$ (equiv., if $w(D) > (1+\delta)w(A)$).
\end{definition}

A feasible set is a \emph{local optimum for $s$-moves} if no $s$-move improves it. A feasible set is a \emph{$\delta$-local optimum for $s$-moves} if no $s$-move $\delta$-improves it. Consider the following local search algorithms for the unweighted and weighted matroid parity problems respectively, parameterized by an integer $s\geq 1$.

\begin{alg} \label{alg1} While there is an $s$-move improving the current solution, initially empty, perform the $s$-move of maximum gain. Return the solution found.
\end{alg}

\begin{alg} \label{alg2} Start from a greedy solution $A_0$ for the problem. While there is an $s$-move $1/n^2$-improving the current solution, perform the $s$-move with maximum gain. Return the solution found.
\end{alg}
\begin{lemma}\label{lem:localoptimum}
Given $s \geq 1$, Algorithm \ref{alg1} (resp.\ Algorithm \ref{alg2}) returns a local optimum (resp.\ a $1/n^2$-local optimum) for $s$-moves for the unweighted (resp.\ weighted) matroid parity problem in $O(n^{s+1})$ time (resp.\ $O(n^{s+2})$ time).
\end{lemma}
\begin{proof}
Clearly, both algorithms are finite and correct and the greedy solution~$A_0$ can be found in $O(n
\log n)$ time, which is $O(n^{s+2})$. Given a feasible solution $A$, the number of possible $s$-moves for~$A$ is $O(n^{s})$. Then each iteration of both algorithms can be done in $O(n^{s})$ time and using the same number of oracle calls. To conclude, we show that the number of iterations is bounded.

Since every iteration of the first algorithm improves the current solution by one unit and the optimum solution has cardinality at most $n$, this algorithm performs at most $n$ iterations. In the second algorithm, every iteration improves the weight of the current solution by a factor $(1+1/n^2)$. Since the initial solution $A_0$ and the solution returned are at most a factor two apart, the number of iterations is at most $\log_{(1+1/n^2)}(2) \leq (1+n^2)\log(2)=O(n^2)$. \qedhere
\end{proof}

The following lemma, which we prove at the end of this section, is useful to obtain a PTAS when $\M$ is strongly base orderable.

\begin{lemma}\label{lem:main}
Let $t\leq n$ be a positive integer and $B$ be an optimal feasible solution of the weighted matroid parity problem on an strongly base orderable matroid.
If $A$ is a local optimum for $(2t+1)$-moves then $w(A) \geq \left(1-1/(t+1)\right)w(B)$.
If $A$ is a $1/n^2$-local optimum for $(2t+1)$-moves then $w(A) \geq \left(1-2/(t+2)\right)w(B)$.
\end{lemma}

Given $0< \varepsilon < 1$, set $t(\varepsilon)=\lceil 1/\varepsilon\rceil$ and $t'(\varepsilon)=2t(\varepsilon)$. By Lemma~\ref{lem:main}, any local optimum for $(2t(\varepsilon)+1)$-moves is a $(1 + \varepsilon)$-approximation, and so is any $1/n^2$-local optimum for $(2t'(\varepsilon)+1)$-moves. By choosing $s$ appropriately in Lemma \ref{lem:localoptimum}, we have the following corollary.

\begin{corollary} Algorithm \ref{alg1} (resp.\ Algorithm \ref{alg2}) returns a $(1+\varepsilon)$-approximation for the unweighted (resp.\ weighted) matroid parity problem for strongly base orderable matroids in time $O(n^{2\lceil1/\varepsilon\rceil+2})$ (resp.\ $O(n^{4\lceil 1/\varepsilon\rceil+3})$).
\end{corollary}

Using the reduction from matroid matching to matroid parity we obtain our main positive result.
\begin{theorem} There is a PTAS for both the weighted and the unweighted matroid matching problems on strongly base orderable matroids.
\end{theorem}

\begin{proof}[Proof of Lemma \ref{lem:main}]
For two feasible solutions $A$ and $B$, consider the collections $A'$ and $B'$ obtained by adding dummy pairs of zero weight to the set of smaller cardinality, so that $|A'|=|B'|\leq n$. Let $X$ be the set of elements added and $\M'=(S\cup X,\I')$ be the matroid obtained from $\M=(S,\I)$ by first adding $X$ as coloops and then truncating $\M'$ to rank $2|A'|=|\bigcup A'|$. By properties of strongly base orderable matroids (see, e.g.~\cite[Section 42.6c]{Schrijver-book}), $\M'$ is also strongly base orderable and both $I=\bigcup A'$ and $J=\bigcup B'$ are bases of $\M'$.

Consider a bijection $\pi\colon I \to J$ such that for all $K \subseteq I$, $\pi(K) \cup (I\setminus K)$ is a base of $\M'$. We can further impose $\pi|_{I \cap J}$ to be the identity map (see, e.g.~\cite[Section 42.6c]{Schrijver-book}). Let $H$ be the multigraph on the set $A' \cup B'$ having one edge between $a=\{u,v\} \in A'\setminus B'$ and $b=\{u',v'\} \in B'\setminus A'$ if only one element in $\{u,v\}$ is mapped to $\{u',v'\}$ and two parallel edges between them if $\pi(\{u,v\})=\{u',v'\}$. Since every vertex has degree two and the graph is bipartite, $H$ is a collection of vertex-disjoint even cycles. Orient the edges so that every cycle is properly directed to obtain a digraph~$\vec{H}$. See Figure \ref{fig:figure1}.

\begin{figure}[h!!]
\centering
\input{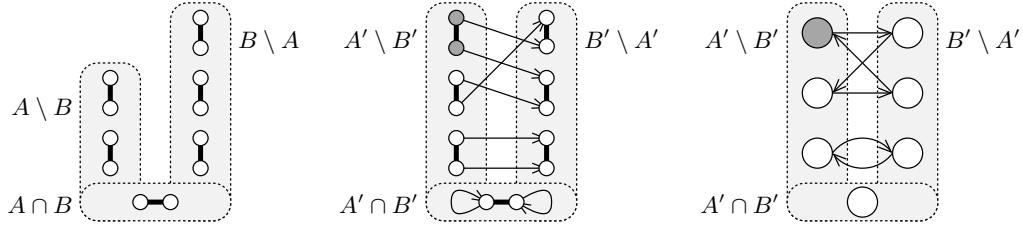}
\caption{In the left, feasible sets $A$ and $B$ sharing one pair. In the middle, sets $A'$ and $B'$ and the bijection $\pi\colon \bigcup A'\to \bigcup B'$. In the right, a possible orientation of $H$.}
\label{fig:figure1}
\end{figure}

For a positive integer $t$, let $\mathcal{C}^+_t$ be the collection of \emph{long cycles} in $\vec{H}$, having length at least $2t+2$ and $\mathcal{C}^-_t$ be the collection of \emph{short cycles} in $\vec{H}$, having length at most $2t$ (for convenience, include in $\mathcal{C}^-_t$ all the isolated nodes in $A'\cap B'$ as \emph{zero-length cycles}). Let $V^+$ and $V^-$ be the sets of nodes in long and short cycles respectively. Note that $V^+$ and $V^-$ form a partition of $A' \cup B'$.

For every cycle $C$ of $\vec{H}$, $\pi(\bigcup (C \cap A')) = \bigcup (C \cap B')$. Hence $A'\Delta C$ is a feasible matching for~$\M'$. Therefore, for every short cycle $C\in \mathcal{C}^-_t$, the collection of pairs $C\setminus X$ defines a $2t$-move for $A$ in the original matroid $\M$.

For every $a \in V^+ \cap A'$, define the set $H_a$ of nodes reachable from $a$ by using at most $2t$ directed arcs in $\vec{H}$. This is, $H_a$ is the node set of the directed path starting in $a$, having $t+1$ nodes from~$A'$ and having $t$ nodes from $B'$. By construction, for every $a\in V^+ \cap A'$, the image under $\pi$ of $\bigcup (H_a \cap A')$ contains $\bigcup (H_a \cap B')$. This observation, together with the definition of $\pi$, implies that $A' \Delta H_a$ is feasible for the matroid $\M'$. Therefore, $H_a \setminus X$ defines a $(2t+1)$-move for $A$ in $\M$.

The total gain of moves coming from short cycles, $\{C\setminus X\colon C \in C^-_t\}$, is
\begin{equation*}
  \sum_{C \in \mathcal{C}^-_t}\![w(C\cap B') - w(C\cap A')]
\end{equation*}
and the total gain of moves from long cycles, $\{H_a\setminus X\colon a \in V^+ \cap A'\}$, is
\begin{equation*}
  \sum_{a \in V^+ \cap A'}[w(H_a \cap B') - w(H_a \cap A')] = \sum_{C \in \mathcal{C}^+_t} [t w(C\cap B') - (t+1)w(C\cap A')].
\end{equation*}

Suppose that $B$ is optimal and $A$ is a local optimum for $(2t+1)$-moves. Then each of the above moves has nonpositive gain. In particular, for each $C \in \mathcal{C}^-_t$, we have $0\geq w(C\cap B') - w(C\cap A') \geq t w(C\cap B') - (t+1)w(C\cap A')$. Therefore,
\begin{align*}
0 &\geq \sum_{C \in \mathcal{C}^-_t\cup \mathcal{C}^+_t}\!\![tw(C\cap B') - (t+1)w(C\cap A')]\\
&= tw(B') -(t+1)w(A') = tw(B)-(t+1)w(A).\end{align*}
From here, we have that $w(A) \geq tw(B)/(t+1)=(1-1/(t+1))w(B)$.

If $A$ is a $1/n^2$-local optimum for $(2t+1)$-moves, then each move has gain of at most $w(A)/n^2$. Since $t\leq n$, $|\mathcal{C}^-_t|\leq |V^-\cap A'|$ and $|A'|\leq n$, we have
\begin{align*}
  w(A) &\geq \frac{w(A)}{n^2}(t|\mathcal{C}^-_t| + |V^+\cap A'|) \geq \sum_{C \in \mathcal{C}^-_t\cup \mathcal{C}^+_t}\!\![tw(C\cap B') - (t+1)w(C\cap A')]\\
&= tw(B') -(t+1)w(A') = tw(B)-(t+1)w(A).
\end{align*}
From here, we have that $w(A) \geq t w(B)/(t+2) = (1 - 2/(t+2)) w(B)$. \qedhere\end{proof}

\section{Matroid parity is hard on strongly base orderable matroids.}
Here we review the proofs of hardness of matroid parity on general matroids (see~\cite[Section 43.9]{Schrijver-book} for more details).
Let $\nu \in \Integ^+$, $S$ be a set of even size and~$E$ be a partition of $S$ into pairs. Let $\M$~be the matroid on $S$, where $T \subseteq S$ is independent if $|T|\leq 2\nu -1$, or if $|T|=2\nu$ and $T$ is not the union of $\nu$ pairs in~$E$. For every set $F \subseteq E$ of size $\nu$, let $\M_F$ be the matroid extending $\M$ by an independent set $I_F=\bigcup F\subseteq S$. Also, given a graph $G=(E,X)$, construct $\M_G$ by extending $\M$ by all independent sets $I_C = \bigcup C$, where $C$ is the vertex set of a clique of size $\nu$ in $G$. It can be shown that $\M$, $\M_F$ and $\M_G$ are matroids, that $E$ has no matching of size $\nu$ feasible for $\M$, that $F\subseteq E$ is the unique matching of size $\nu$ feasible for $\M_F$, and that $E$ contains a matching of size $\nu$ feasible for $\M_G$ if and only if $G$ has a clique of size $\nu$.

Given $\nu$, the above set $E$ and an oracle for a matroid $\mathcal{N}$ that is known to be either $\M$ or any of the possible $\M_F$ with $|F|=\nu$, it is known that the number of oracle calls needed to check if there is a matching of size $\nu$ feasible for $\mathcal{N}$ is at least $\binom{|E|}{\nu}$, showing that the matroid parity problem is not in oracle co-NP. Similarly, as the maximum-size clique problem is NP-complete, the matroid parity problem on matroids $\M_G$ constructed as above is NP-complete.

Note that $\M$ and $\M_F$ are special cases of $\M_G$, obtained from a graph having zero or one clique of cardinality $\nu$. The following lemma implies that the hardness guarantees for the matroid parity problem are maintained if we restrict to strongly base orderable matroids.

\begin{lemma}\label{lem:lemsbo}
  For every $\nu \in \Integ^+$ and every graph $G$, the matroid $\M_G$ is strongly base orderable.
\end{lemma}
\begin{proof}
Let $\nu \in \Integ^+$ and $G=(E,X)$ be a graph. Create a set $S$ having two copies of each element in $E$, and interpret $E$ as a partition of $S$ into pairs or mates. The bases of $\M_G$ are then
\begin{equation*}\begin{split}
\mathcal{B}(\M_G) =\{I\subseteq S\colon |I|= 2\nu, \text{ and $I$ is not the union of $\nu$ pairs of $E$}\}\\
\cup\, \{I\subseteq S\colon \text{$I$ is the union of $\nu$ pairs $F$, where $F$ is a clique in $G$}\}.
\end{split}
\end{equation*}

Let $I,J \subseteq S$ be two bases in $\mathcal{B}(\M_G)$. In what follows we construct a bijection $\pi\colon I \to J$, such that for all $K \subseteq I$, $L=\pi(K) \cup (I\setminus K)$ is also a base. We divide the problem in cases. In each of them we impose conditions on $\pi$ that can easily be satisfied by some bijection. In particular, for every case, we impose $\pi|_{I \cap J}$ to be the identity map. By assuming that $L$ is the union of $\nu$~pairs, we either show that $L \in \{I,J\}$ or we get a contradiction, proving in both situations that $L$ is a base.

Under the assumption that $L$ is the union of $\nu$ pairs, the following two properties hold.
\begin{description}
\item[Property 1:] An element is in $L$ if and only if its mate is in $L$.
\item[Property 2:] An element in $I\Delta J$ is in $L$ if and only if its image or preimage under $\pi$ is not in $L$.
\end{description}

There are nine cases to consider for $I$ and $J$. For visual aid on the first five cases, see Figure \ref{fig:figure2}.

\begin{figure}[h!!t]
\centering
\scalebox{0.98}{\input{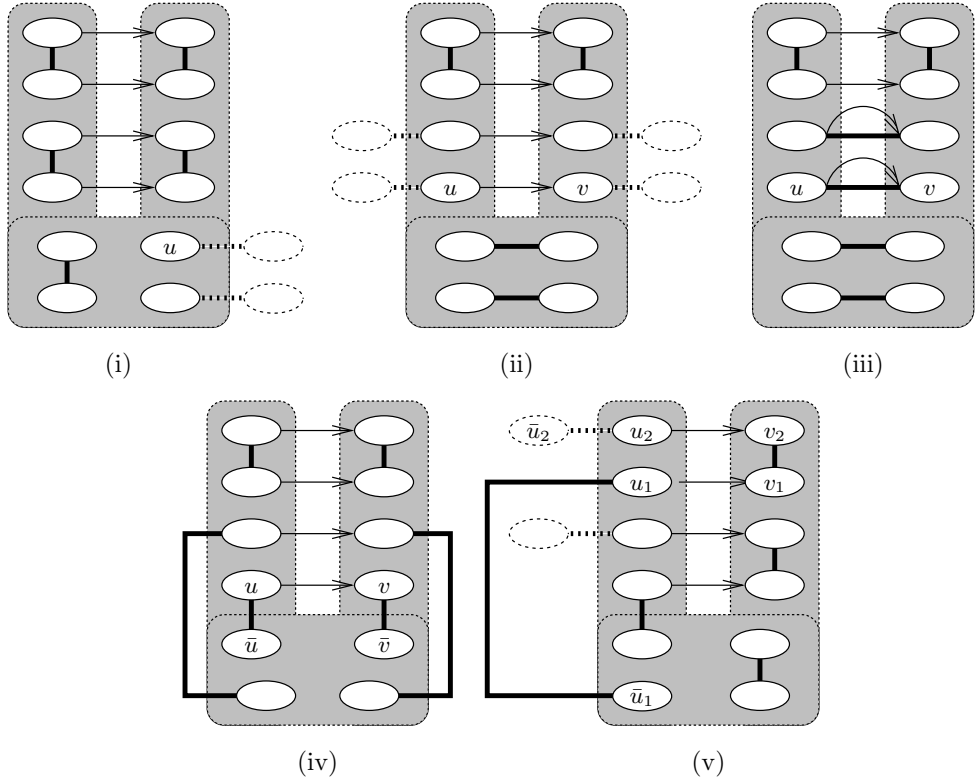}}
\caption{Cases (i)--(v). In each case, $\nu=4$, the top-left region is $I\setminus J$, the top-right region is $J\setminus I$ and the bottom region is $I\cap J$. Thick lines represent $E$.  Thin arrows represent a possible map $\pi\colon I \to J$.}
\label{fig:figure2}
\end{figure}

\begin{enumerate}
\item[(i)] If there is $u \in I\cap J$ whose mate is in $S\setminus(I\cup J)$, then any bijection $\pi$ works.
\end{enumerate}
Since $\pi(u)=u$, any set $L=\pi(K) \cup (I\setminus K)$ as above contains $u$ but not its mate. Thus, $L$ is not the union of $\nu$ pairs.
\begin{enumerate}
  \item[(ii)] Else, if there are $u\in I\setminus J$ and $v\in J\setminus I$ with mates in $S\setminus(I \cup J)$, then select any such $u$ and $v$, and set $\pi(u)=v$.
  \item[(iii)] Else, if there are $u\in I\setminus J$ and $v\in J\setminus I$ with $\{u,v\}\in E$, then select any such pair and set $\pi(u)=v$.
\end{enumerate}
Cases (ii) and (iii) are similar: any set $L=\pi(K) \cup (I\setminus K)$ contains either $u$ or $v$, but not its corresponding mate. Thus, $L$ is not the union of $\nu$ pairs.
\begin{enumerate}
\item[(iv)]\label{iv} Else, if there are pairs $\{u,\bar{u}\}$ and $\{v,\bar{v}\}$ in $E$ with $u\in I\setminus J$, $v\in J\setminus I$ and $\{\bar{u},\bar{v}\}\subseteq I\cap J$, then set $\pi(u)=v$.
\end{enumerate}
For case (iv), note that any set $L$ as above contains both $\bar{u}$ and $\bar{v}$ but only one of their mates. Thus,  $L$ is not the union of $\nu$ pairs.
\begin{enumerate}
\item[(v)] Else, if there are pairs $\{u_1,\bar{u}_1\},\{u_2,\bar{u}_2\},\{v_1,v_2\}$ in $E$ with $u_i \in I\setminus J$, $v_i \in J\setminus I$, for $i\in\{1,2\}$,  $\bar{u}_1 \in I\cap J$, and $\bar{u}_2 \in S\setminus(I\cup J)$, then set $\pi(u_i)=v_i$, for $i\in\{1,2\}$. (Also consider here the same case with $I$ and $J$ interchanged.)
\end{enumerate}
For case (v), consider a set $L$ as above and assume $L$ is the union of $\nu$ pairs. Since $\bar{u}_1$ is in $I\cap J$, we have $\pi(\bar{u}_1)=\bar{u}_1$, and so $\bar{u}_1$ must be in $L$. In what follows we apply Properties 1 and 2 iteratively. Since $\bar{u}_1$ is in $L$, its mate $u_1$ is also in $L$. Then, $\pi(u_1)=v_1$ is outside set $L$, and so is its mate $v_2=\pi(u_2)$. This implies that $u_2$ is in $L$. But its mate $\bar{u}_2$ is outside $L$, contradicting the assumption.

If none of the previous cases hold, then only one of $I\setminus J$ or $J\setminus I$ can contain an element with mate in $I\cap J$, and  only one  can contain an element with mate in $S\setminus (I \cup J)$. Furthermore, none of the two sets can contain both types of elements. The rest of the pairs intersecting $I\cup J$ are completely inside $I\setminus J$, $J\setminus I$ or $I\cap J$. Therefore, by possibly interchanging the roles of $I$ and~$J$, we can assume that (1) the set $I\setminus J$ contains $a\geq 0$ elements with mates in $I\cap J$, and $k\geq 0$ pairs completely inside; (2) the set $J\setminus I$ contains $b\geq 0$ elements with mates in $S\setminus (I \cup J)$ and $\ell \geq 0$ pairs completely inside; (3) the set $I \cap J$ contains $c\geq 0$ pairs completely inside; and (4) no pair other than the ones described before intersects $I\cup J$.

Under these assumptions, let $(u_1,u'_1),\dots,(u_k,u'_k)$ and $(v_1,v'_1),\ldots,(v_\ell,v'_\ell)$ be the pairs completely inside $I\setminus J$ and $J\setminus I$ respectively. Note that $2k + a = 2\ell + b=|I\setminus J|=|J\setminus I|$. In particular, if $b > a$, then  $k > \ell$ and viceversa. We proceed with the rest of the cases. For visual aid, see Figure~\ref{fig:figure3}.

\begin{figure}[h!!t]
\centering
\scalebox{0.98}{\input{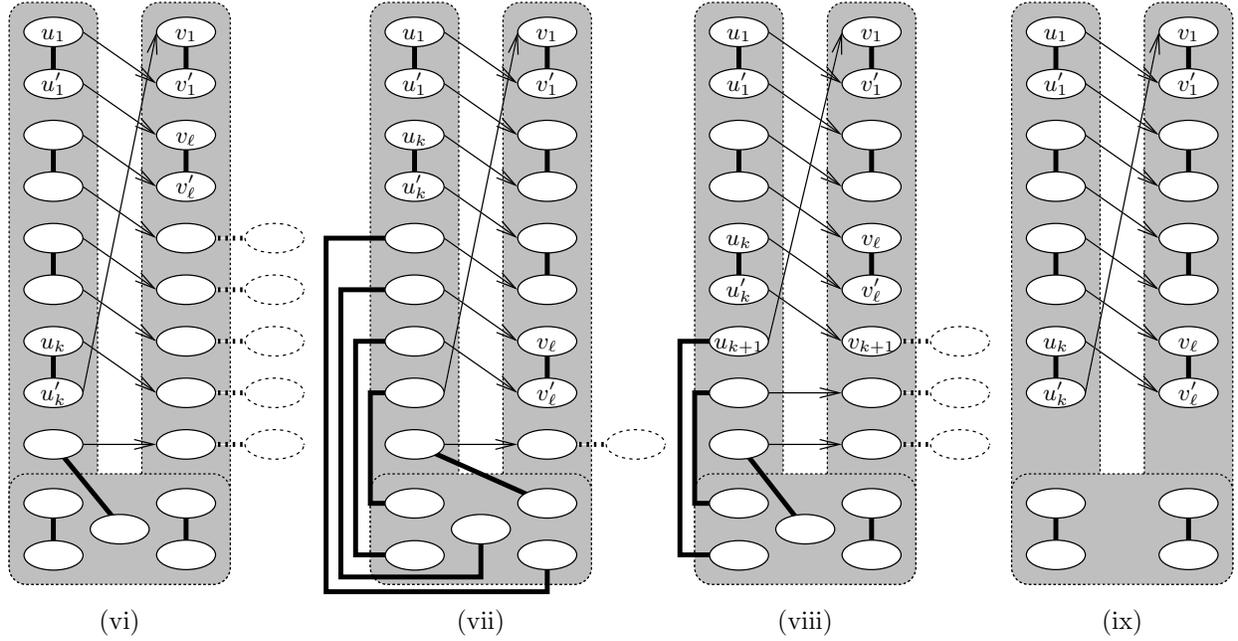}}
\caption{Cases (vi)--(ix), corresponding to $(a,b,\nu)=(1,5,7);(5,1,7);(3,3,7);(0,0,6)$. In each case, the top-left region is $I\setminus J$, the top-right region is $J\setminus I$ and the bottom region is $I\cap J$.
Thick lines represent $E$.  Thin arrows represent a possible map $\pi\colon I \to J$. The elements of $I\setminus J$ and $J\setminus I$ are numbered from up to down as $u_1,u'_1,u_2,u'_2,\dots$ and $v_1,v'_1,v_2,v'_2,\ldots$, respectively.}
\label{fig:figure3}
\end{figure}

\begin{enumerate}
\item[(vi)] If $a < b$, let $\{v_{i},v'_{i}\colon i = \ell+1,\dots,k\}$ be a set of $2(k-\ell)=b-a$ different elements of $J\setminus I$ with mates in $S\setminus (I\cup J)$. Set $\pi(u_i)=v'_i$ and $\pi(u'_i)=v_{i+1\!\pmod k}$ for every $1\leq i \leq k$.
\end{enumerate}
Assume that $L$ is the union of $\nu$ pairs. Then no element with mate in $S\setminus (I \cap J)$ is in $L$. Using Property 2, the preimages under $\pi$ of these elements are in~$L$. In particular, $u_{\ell+1},\dots, u_k$ are in~$L$. By repeatedly using Properties~1~and~2, starting on element $u_k$, we conclude that all $u_1,u'_1,\ldots,u_k,u'_k$ are in $L$, and that their images, $v_1,v'_1,\ldots,v_k,v'_k$ are not in $L$. Furthermore, since the $a$~elements of $I$ with mates in $I\cap J$ are also in $L$, we conclude that $L=I$.

\begin{enumerate}
\item[(vii)] If $b < a$, let $\{u_{i},u'_{i}\colon i = k+1,\dots,\ell\}$ be a set of $2(\ell-k)=a-b$ different elements of $I\setminus J$ with mates in $I\cap J$. Set $\pi(u_i)=v'_i$ and $\pi(u'_i)=v_{i+1\!\pmod \ell}$ for every $1\leq i \leq \ell$.
\end{enumerate}
Assume that $L$ is the union of $\nu$ pairs. Then every element with mate in $I\cap J$ is in $L$. Using Property~2, the images of these elements are not in~$L$. This is, the elements $v'_{k+1}, v_{k+2}, v'_{k+2},\dots, v_\ell, v'_\ell, v_1$ are not in $L$. By repeatedly using Properties 1 and 2, starting on element $v_1$, we conclude that $L=I$.
\begin{enumerate}
\item[(viii)] If $b = a > 0$, let $u_{k+1}$ be an element in $I\setminus J$ with mate in $I\cap J$, and $v_{k+1}$ be an element in $J\setminus I$ with mate in $S\setminus(I\cup J)$. Set $\pi(u_i)=v'_i$, $\pi(u'_i)=v_{i+1}$ for every $1\leq i \leq k$ and $\pi(u_{k+1})=v_1$.
\end{enumerate}
If $L$ is the union of $\nu$ pairs, then every element with mate in $I\cap J$ is in $L$. In particular, $u_{k+1} \in L$ and, by Property 2, its image $v_1$ is outside $L$. By repeatedly using Properties 1 and 2, starting on element $v_{1}$ we get $L=I$. This proof also holds in the degenerated case where $k=\ell=0$.

\begin{enumerate}
  \item[(ix)] If $b=a=0$. Set $\pi(u_i)=v'_i$, $\pi(u'_i)=v_{i+1\!\pmod k}$ for every $1\leq i \leq k$.
\end{enumerate}
For this final case, let $L$ be a union of $\nu$ pairs. If $u_1 \in L$ then it is easy to see, using Properties 1 and 2, that $L=I$. Similarly, if $u_1 \not\in L$ then $L=J$.

We have proved that for every case, and for every $K\subseteq I$, if the set $L=\pi(K) \cup (I \setminus K)$ is a union of pairs then $L\in \{I,J\}$. Thus, $L$ is always a base of $\M_G$, and therefore $\M_G$ is an strongly base orderable matroid. \qedhere
\end{proof}

\paragraph{Acknowledgment} The author would like to thank Jan Vondr\'ak for helpful discussions.

\bibliographystyle{abbrv}
\bibliography{biblioMatroidMatching}

\end{document}